\documentclass[11pt]{article}
\usepackage[english]{babel}
\usepackage{latexsym}
\usepackage{amsmath}
\usepackage{amssymb}
\usepackage{amsthm}
\usepackage{epsfig}
\everymath=\expandafter{\the\everymath\displaystyle}
\DeclareMathSizes{24}{24}{24}{24}

\newtheorem{theorem}{Theorem}

\newtheorem{lemma}[theorem]{Lemma}

\newtheorem{definition}{Definition}

\advance \topmargin by -\headheight
\advance \topmargin by -\headsep
\evensidemargin \oddsidemargin
\begin{document}

\title{How many subsets of edges of a directed multigraph can be represented as trails?}
\author{Joseph Shayani\footnote{I am indebted to Michael Ostrovsky for presenting me with the problem of interest in this paper.}}

\maketitle
\begin{abstract}
For each subset of edges of a (directed multi-) graph, one may determine whether the edges can be represented as a trail. I prove that the fraction trail-representable subsets of edges is at most $O(\sqrt{\log m}/\sqrt{m})$, where $m$ is the number of edges, and show by example that the upper bound is nearly tight.
\end{abstract}

\section{Introduction}
For each subset of edges of a graph, one can determine whether it is possible to arrange the edges in some order such that the resulting sequence of edges is a trail. Thus, for any graph one may count the number of subsets of edges representable as trails and consider the fraction of the total number of subsets of edges ($2^m$ where $m$ is the number of edges) that are representable as trails. 

One may suspect that the trail-fraction must become small as the number of edges grows large. In this paper, I confirm this suspicion. I prove that the trail-fraction of a graph with $m$ edges is bounded by $O(\sqrt{\log {m}} / \sqrt{m})$. Moreover, I show that this upper bound is nearly tight by demonstrating a family of graphs with trail-fraction at least $\Omega(1/\sqrt{m})$.

The trail-fraction may be of interest, for example, when a graph represents a network of buyers and sellers. In such a setting, Hatfield et al. \cite{hatfield2015chain} prove that the matching-theoretic property of stability, which \textit{a priori} is a condition over all subsets of edges, is actually equivalent to the property ``chain stability," a condition over only subsets of edges representable as a trail or ``chain." The trail-fraction may be interpreted as the savings generated by this equivalence.

\section{Setup}\label{intro}
I begin by giving the definitions of multigraphs and trails, and then I define the trail-fraction.

\begin{definition} 
A \emph{directed mutligraph} $G=(V,E)$ is a set of vertices $V$ together with a set of edges $E$. Each edge $e\in E$ has an origin $s(e)$ and end $t(e)$ in $V$ with $s(e) \ne t(e)$. In a multigraph, distinct edges $e\ne e'$ \textit{may} have the same origin $s(e)=s(e')$ and end $t(e)=t(e')$. 
\end{definition}
I will always mean ``directed multigraph" when I write ``graph."

\begin{definition} 
A subset of edges $T \subset E$ can be represented as a \emph{trail} if there exists an ordering $e_1, e_2, \dots, e_k$ of edges in $T$ (choosing each exactly once) such that $t(e_i)=s(e_{i+1})$ for all $i = 1, \dots, k-1$.
\end{definition}

Note that my definition of trail allows for repetition of vertices (but not repetition of edges, because I am considering subsets of edges). 

Consider the set $\Gamma(m)$ of all graphs $G=(V,E)$ with edge set of size $|E| = m$. Let $d(G)$ be the number of distinct subsets $T \subset E$ that can be represented as a trail. Define the trail-fraction
$$f(G) = d(G)/2^m.$$

\section{Upper bound for trail-fraction}
\begin{theorem}\label{thm}
Let $f(G)$ be the trail-fraction as defined in Section \ref{intro}. Then
$$f(G) \le \sqrt{\log{m}}/\sqrt{m},$$
where $m$ is the number of edges in $G$.
\end{theorem}

The intuition for the proof is the following: Imagine that for each edge $e$ in a graph $G$, I independently flip a fair coin in order to decide whether to include $e$ in subset $T$. Then  $d(G) / 2^m$ is exactly the probability that $T$ is representable as a trail. $G$ must have either a vertex $v$ with high degree or a large number of vertices.\footnote{I am indebted to Jan Vondrak for his suggestion to examine two cases of this sort in the proof of Theorem \ref{thm}.} In order for $T$ to be a trail, each vertex in $T$ must have balanced (within $\pm 1$) in and out degrees. The probability of balancing a node with large degree will be small, as will be the probability of balancing many nodes simultaneously (even if these nodes have small degree). Hence in any case, the probability that $T$ is a trail will be small.

Now I present the formal argument. I begin by introducing notation to handle the notion of the degree of a vertex and then make precise the observation that balanced in- and out- degrees are a necessary condition for trails. 
\begin{definition}
The \emph{in-degree} $\delta^-(v, G_1)$ (resp. \emph{out-degree} $\delta^+(v, G_1)$) of a vertex $v\in V$ with respect to subset $G_1=(V_1, E_1)$ of graph $G=(V,E)$ is the number of inbound edges $|\{e\in E_1: v=t(e)\}|$ (resp. outbound edges $|\{e\in E_1: v=s(e)\}|$). The \emph{total degree} $\delta(v, G_1)$ is equal to the sum of in- and out-degrees $\delta^-(v, G_1) + \delta^+(v, G_1)$.
\end{definition}

For a trail $T\subset G$ of size $k$, the in- and out-degrees with respect to $T$ of each vertex in $T$  (i.e. each vertex $v$ for which there exists edge $e\in T$ with $s(e)=v$ or $t(e)=v$) must be equal, with the exception of perhaps the nodes $s(e_1), t(e_k)$. In fact, the same condition must hold for all but two vertices in $V$ (because the vertices not in $T$ will have in- and out-degrees $0$ with respect to $T$). 

Next we introduce the notion of a edge-increasing sequence of vertices and prove a lemma that gives a lower bound on the size of such a sequence we can find in any graph.

\begin{definition}
For a subset of vertices $U\in V$ in graph $G=(V,E)$, write $I(U)$ for the set of edges incident to at least one $v\in U$, or simply $I(v)$ if $U$ is a singleton $\{v\}$. Define an \emph{edge-increasing sequence} of length $r$ to be a sequence of vertices $v_1, \dots, v_r$ with the property that $I(v_{i+1}) \not\subset I(\{v_1,\dots, v_i\})$ for each $i = 1,\dots, r-1$.
\end{definition}

The edge-increasing property will be exactly what I need in order to drive down the cumulative probability of balancing in- and out-degrees at each vertex in the case that $G$ has many vertices. Note that an independent set of vertices is a special case of an edge-increasing sequence.

\begin{lemma}
In graph $G=(V,E)$, if $I(v) \ge 1$ for every $v\in V$, then there exists an edge-increasing sequence of length $|V|/2$.
\end{lemma}
\begin{proof} Initialize $S\subset V$ to be the empty set. Each time I add a vertex $v$ to $S$, set 
$$G = (V \backslash \{v\}, E \backslash I(v)),$$
and further remove from $V$ the set of vertices with no remaining incident edges after removing $I(v)$. If I iteratively add to $S$ the vertex $v$ with the fewest remaining incident edges, I eliminate from $V$ at most $2$ vertices: $v$ and at most one other. Indeed, if upon adding $v$ and removing edges $I(v)$, I leave some other vertex $u$ with no more edges, then it must have been true that $I(u)\subset I(v)$. Since $v$ had the smallest set of incident edges, it must have been that $I(u)= I(v)$, i.e. I removed only edges between $v$ and $u$. Hence I can increase the size of $S$ at least $|V|/2$ times. 
\end{proof}

\begin{proof}[Proof of Theorem \ref{thm}] Fix graph $G=(V,E)$ with $m$ edges. Let $k$ be a number I will choose later. If no vertex $v\in V$ has degree $\delta(v, E) \ge k$, then by the pigeonhole principle, $V$ contains at least $2m / k$ distinct vertices $v$ with $|I(v)| \ge 1$, and by the previous lemma, $V$ has an edge-increasing sequence of size $m/k$.  To summarize, one the following two holds:
\begin{enumerate}
\item There exists an vertex $v\in V$ of degree $\delta(v,E)$ at least $k$.
\item There exists an edge-increasing sequence of size at least $m/k$.
\end{enumerate}

In the first case, either the in-degree or out-degree of $v$ with respect to $E$ is at least $k/2$, without loss of generality the in-degree. Consider flipping a coin for each edge in $E$ for inclusion in $T$, and condition on the event that $\delta^+(v, T) = j  \in [0,1,\dots, k/2]$. It is necessary that in order for $T$ to be a trail $\delta^+(v, T)  \in \{ j-1, j, j+ 1\}$. Let $c = \delta^-(v, E)$, and let $X$ be the random variable $\delta^-(v, T)$ representing the number of edges selected for inclusion in T. Then,
\begin{align*}
P(X \in \{j-1, j, j + 1\}) &= (1/2)^c \left( \binom{c}{j-1} + \binom{c}{j} + \binom{c}{j+1}\right) \\
&\le 3 (1/2)^c \binom{c}{c/2} \\
&= O(1/\sqrt{c}),
\end{align*}
where we used the Stirling approximation to handle the binomial factor.\footnote{Stirling's bound states that $\sqrt{2\pi}n^{n+1/2}e^{-n} \le n! \le en^{n+1/2}e^{-n}$. Hence $$\binom{c}{c/2} = \frac{c!}{(c/2)!^2} \le \frac{ec^{c+1/2}e^{-c} }{(\sqrt{2\pi}(c/2)^{c/2+1/2}e^{-c/2})^2} = \frac{e}{2\pi c^{1/2} (1/2)^{c+1}} = \frac{e}{\pi \sqrt{c} (1/2)^c}.$$ \label{fn}} Now, $c \ge k/2$, so $f(G) = O(1/\sqrt{k})$. 

In the second case, write $r = m/k$ for the size of the edge-increasing sequence of vertices $S = \{v_1, \dots, v_r\}$. As I observed, in order for $T$ to be a trail, it is necessary that $\delta^+(v, T) = \delta^-(v, T)$ for all but two vertices $v\in S$. For each $i=1, \dots, r$, choose an edge $e_i \in I(v_i) \backslash I(\{v_1,\dots, v_{i-1}\})$. Conditional on all coin flips up on all edges $I(\{v_1,\dots, v_{i-1}\}) \backslash \{e_i\}$, the probability over coin flip on $e_i$ that $\delta^+(v_i, T) = \delta^-(v_i, T)$ is at most $1/2$. Let $X$ be the random variable equal to the number of values of $i$ for which $\delta^+(v_i, T) = \delta^-(v_i, T)$. Then, 
\begin{align*}
P(X \ge r-2) &\le \binom{r}{2} (1/2)^{r - 2}  + r (1/2)^{r - 1} + (1/2)^r\\
&\le (2 r^2 + 2r +1) (1/2)^r \\
&\le 4r^2 (1/2)^r 
\end{align*}
for $r > 1$. Taking $r = \log_2{m}$ yields maximal probability $4(\log_2{m})^2/m$, and this requires us to take $k = m/\log_2{m}$, which yields maximal probability $O(\sqrt{\log{m}}/\sqrt{m})$ in the first case. Asymptotically, this first case probability dominates. \end{proof}

\section{Lower bound}

In this section, I demonstrate a family of graphs $\{G(m)\}$ with $f(G(m))= \Omega(1/\sqrt{m})$. In particular, the upper bound presented in the previous section is tight up to a square-root logarithmic term. 

Consider the graph $G(m)=(V,E(m))$ with two vertices $V=\{v_1, v_2\}$ and $m/2$ edges $e$ each with $s(e)=v_1, t(e)=v_2$ and vice versa. I will count the even-sized subsets of $E$ that are representable as trails. The number of trail-representable subsets of size $2\ell$ is exactly the number of way to choose $\ell$ edges from the first set of $m/2$ edges (in one direction) times the number of way to choose $\ell$ edges from the second set of $m/2$ edges (in the other direction). So the total number of even length trail-representable subsets is
$$\sum_{\ell = 1}^{m/2} \binom{m/2}{\ell}^2 = \binom{m}{m/2},$$
where I used a counting identity.\footnote{Consider a set $S$ of $m$ distinct elements, and choose an arbitrary division of $S$ into disjoint subsets $A$ and $B$, each with $m/2$ elements. Each distinct choice of $m/2$ elements from $S$ corresponds uniquely to a choice of $\ell$ elements from $A$ and $m/2 - \ell$ from $B$ for exactly one value of $\ell$. Hence $$\binom{m}{m/2} = \sum_{\ell = 1}^{m/2} \binom{m/2}{\ell}\binom{m/2}{m/2 - \ell} = \sum_{\ell = 1}^{m/2} \binom{m/2}{\ell}^2.$$} 
The Stirling approximation\footnote{See footnote \ref{fn}.} yields asymptotic lower bound $d((G(m)) =\Omega({2^m}/{\sqrt{m}})$, so $f(G(m)) = \Omega(1/\sqrt{m})$.


\begin{thebibliography}{Ost08}

\bibitem[1]{hatfield2015chain}
John William Hatfield, Scott Duke Kominers, Alexandru Nichifor, Michael Ostrovsky, Alexander Westkamp.
\newblock Chain Stability in Trading Networks.
\newblock Working paper, April 2015.

\end{thebibliography}

\end{document}